\documentclass[a4paper,english]{revtex4}
\usepackage{libertine-type1}
\usepackage[T1]{fontenc}
\usepackage[latin9]{inputenc}
\usepackage{amsthm}
\usepackage{amsmath}
\usepackage{amssymb}

\makeatletter

\@ifundefined{textcolor}{}
{%
 \definecolor{BLACK}{gray}{0}
 \definecolor{WHITE}{gray}{1}
 \definecolor{RED}{rgb}{1,0,0}
 \definecolor{GREEN}{rgb}{0,1,0}
 \definecolor{BLUE}{rgb}{0,0,1}
 \definecolor{CYAN}{cmyk}{1,0,0,0}
 \definecolor{MAGENTA}{cmyk}{0,1,0,0}
 \definecolor{YELLOW}{cmyk}{0,0,1,0}
}
\numberwithin{equation}{section}
\numberwithin{figure}{section}
  \theoremstyle{plain}
  \newtheorem{lem}{\protect\lemmaname}
\theoremstyle{plain}
\newtheorem{thm}{\protect\theoremname}

\makeatother

\usepackage{babel}
  \providecommand{\lemmaname}{Lemma}
\providecommand{\theoremname}{Theorem}

\begin{document}

\title{Entanglement sharing via qudit channels: Nonmaximally entangled states
may be necessary for one-shot optimal singlet fraction and negativity}

\author{Rajarshi Pal }

\email{rajarshipal4@gmail.com}

\affiliation{Department of Physics, Indian Institute of Technology Madras, Chennai
600036}

\author{Somshubhro Bandyopadhyay}

\email{som@jcbose.ac.in}

\affiliation{Department of Physics and Center for Astroparticle Physics and Space
Science, Bose Institute, EN-80, Sector V, Bidhannagar, Kolkata 700091}
\begin{abstract}
We consider the problem of establishing entangled states of optimal
singlet fraction and negativity between two remote parties for every
use of a noisy quantum channel and trace-preserving LOCC under the
assumption that the parties do not share prior correlations. We show
that for a family of quantum channels in every finite dimension $d\geq3$,
one-shot optimal singlet fraction and entanglement negativity are
attained only with appropriate nonmaximally entangled states. A consequence
of our results is that the ordering of entangled states in all finite
dimensions may not be preserved under trace-preserving LOCC. 
\end{abstract}
\maketitle

\section{Introduction }

In quantum information theory, entangled states \cite{entanglement-schrodinger,Entanglement-horodecki}
shared between remote parties are considered as resources \cite{Entanglement-horodecki}
within the paradigm of local operations and classical communication
(LOCC) (see for example, \cite{Chitambar-LOCC}). However, any protocol
of entanglement sharing requires sending quantum systems over quantum
channels along with local processing irrespective of preshared correlations
that may be present between the parties \cite{Bennett-1996,Bennett-purification-1996,Deutsch-1996,Horodecki-1997,Horodecki-general-teleportation-1999,Schumacher-1996,Streltsov-etal-2015,Zuppardo-et-al-2016}.
It may be noted that recent results \cite{Streltsov-etal-2015} strongly
suggest that protocols with prior correlations may not provide any
efficiency advantage over the ones without correlations. 

In this paper we consider the basic protocol between two remote parties,
Alice and Bob, who do not share any prior correlation. Such a protocol
may be described as follows. Alice locally prepares a pure quantum
state $\left|\psi\right\rangle \in\mathbb{C}^{d}\otimes\mathbb{C}^{d}$
and sends half of it to Bob down a $d$-dimensional quantum channel
$\Lambda$. In an ideal scenario where the channel is taken to be
noiseless, maximally entangled states are easily established this
way. For a noisy channel, which is typically the case, Alice and Bob
end up with a mixed state $\rho_{\psi,\Lambda}=\left(\mathcal{I}\otimes\Lambda\right)\rho_{\psi}$
where $\rho_{\psi}=\left|\psi\right\rangle \left\langle \psi\right|$
or an ensemble for many uses of the channel. Thus in a noisy channel
scenario the goal is to establish entangled states that are optimal
with respect to some well-defined figure of merit. Entanglement distillation
\cite{Bennett-1996,Bennett-purification-1996,Deutsch-1996,Horodecki-1997}
provides a solution by converting many copies of $\rho_{\psi,\Lambda}$
to fewer near-perfect entangled states thereby requiring many uses
of the channel and joint measurements. 

The present paper considers a one-shot instance of the entanglement
sharing problem where the goal is to establish entangled states of
maximum singlet fraction and entanglement negativity \cite{Negativity}
achievable for every single use of the channel (see for example, \cite{BG-2012,PBG-2014}).
As we will see, the one-shot optimal values of these two quantities
are closely related and exhibit similar properties. 

The singlet fraction (or maximally entangled fraction) \cite{Bennett-1996,Deutsch-1996,Bennett-purification-1996,Verstrate-fidelity-2002,Horodecki-general-teleportation-1999}
of $\rho_{\psi,\Lambda}$ is given by 
\begin{eqnarray}
\mathbb{F}\left(\rho_{\psi,\Lambda}\right) & = & \max_{\vert\Phi\rangle}\langle\Phi\vert\rho_{\psi,\Lambda}\vert\Phi\rangle\label{fidelity-1}
\end{eqnarray}
where $|\Phi\rangle$ is a maximally entangled state in $\mathbb{C}^{d}\otimes\mathbb{C}^{d}$.
The motivation behind choosing singlet fraction as our figure of merit
lies in the fact that singlet fraction is an effective measure of
usefulness of the state $\rho_{\psi,\Lambda}$ for quantum information
processing tasks, e.g. quantum teleportation \cite{Teleportation},
superdense coding \cite{dense-coding-1}, quantum key distribution
\cite{crypto-ekert}, and distributed computation \cite{distributed-computation},
which typically require entangled states of very high $\mathbb{F}$,
ideally close to unity. It is useful to note that the yield in a distillation
protocol depends on $\mathbb{F}\left(\rho_{\psi,\Lambda}\right)$;
in fact, for the distillation protocols to work the singlet fraction
of the mixed states must exceed a certain threshold value \cite{Bennett-1996,Deutsch-1996,Bennett-purification-1996}. 

While one may suppose that maximizing $\mathbb{F}\left(\rho_{\psi,\Lambda}\right)$
given by Eq.$\,$(\ref{fidelity-1}) over all transmitted states $\vert\psi\rangle$
will yield the desired result, such a supposition may be unfounded.
This is because singlet fraction of a state can increase under local
trace-preserving operations (TP-LOCC) \cite{Badziag-2000,Bandyopadhyay-2002,Verstrate-fidelity-2002}
which strongly suggests that in a one-shot protocol local post-processing
may be required to attain the optimal value. Taking this into account,
let $\mbox{\ensuremath{\rho}}_{\psi,\Lambda}^{{\rm L}}={\rm L}\left(\rho_{\psi,\Lambda}\right)$
denote the density matrix under the action of some TP-LOCC operation
${\rm L}$ on $\rho_{\psi,\Lambda}$. Then, for a fixed transmitted
state $\left|\psi\right\rangle $, the maximum achievable singlet
fraction is defined as \cite{Verstrate-fidelity-2002}
\begin{eqnarray}
\mathbb{F}^{*}\left(\rho_{\psi,\Lambda}\right) & = & \max_{{\rm L\in TP-LOCC}}\mathbb{F}\left(\rho_{\psi,\Lambda}^{{\rm L}}\right),\label{max-achievable-fidelity-1}
\end{eqnarray}
where the maximization is over all TP-LOCC ${\rm L}$. Note that,
unlike $\mathbb{F}$ which can increase under TP-LOCC, $\mathbb{F}^{*}$
is a LOCC monotone \cite{Verstrate-fidelity-2002}. It is important
to note that the action of optimal TP-LOCC, say, ${\rm L}^{*}$ on
$\rho_{\psi,\Lambda}$ results in a density matrix, say $\rho_{\psi,\Lambda}^{\prime}={\rm L}^{*}\left(\rho_{\psi,\Lambda}\right)$.
Thus, we can write 
\begin{eqnarray}
\mathbb{F}^{*}\left(\rho_{\psi,\Lambda}\right) & = & \mathbb{F}\left(\rho_{\psi,\Lambda}^{\prime}\right).\label{F*=00003DFrho*}
\end{eqnarray}
The one-shot optimal singlet fraction for the channel $\Lambda$ is
defined as \cite{PBG-2014} 
\begin{eqnarray}
\mathbb{F}\left(\Lambda\right) & = & \max_{\left|\psi\right\rangle }\mathbb{F}^{*}\left(\rho_{\psi,\Lambda}\right),\label{optimal-fidelity}
\end{eqnarray}
where the maximum is taken over all pure state transmissions. Let
us now suppose that $\left|\psi_{{\rm opt}}\right\rangle $ is a pure
entangled state such that (\ref{optimal-fidelity}) holds; then, 
\begin{eqnarray}
\mathbb{F}\left(\Lambda\right) & = & \mathbb{F}^{*}\left(\rho_{\psi_{{\rm opt}},\Lambda}\right)=\mathbb{F}\left(\rho_{\psi_{{\rm opt}},\Lambda}^{\prime}\right).\label{optimal-purity-one}
\end{eqnarray}
The one-shot optimal singlet fraction is related to optimal negativity
in the following way. For any two-qudit density matrix $\sigma$ the
following inequality holds \cite{VH-2003}: 
\begin{eqnarray}
\mathbb{F}^{*}\left(\sigma\right) & \leq & \frac{1+2\mathcal{N}\left(\sigma\right)}{d}\label{F*<=00003Dnegativity}
\end{eqnarray}
where $\mathcal{N}\left(\sigma\right)$ denotes the negativity \cite{Negativity}
of the state $\sigma$. Now, substituting $\sigma$ by $\rho_{\psi,\Lambda}$
in the above inequality and maximizing over all transmitted states
$\left|\psi\right\rangle $ leads to an upper bound on $\mathbb{F}\left(\Lambda\right)$:
\begin{eqnarray}
\mathbb{F}\left(\Lambda\right) & \leq & \frac{1+2\mathcal{N}\left(\rho_{\psi_{{\rm opt}},\Lambda}\right)}{d}\leq\frac{1+2\mathcal{N}\left(\Lambda\right)}{d},\label{upper bound on optimal SF}
\end{eqnarray}
where $\mathcal{N}\left(\Lambda\right)=\max_{\left|\psi\right\rangle }\mathcal{N}\left(\rho_{\psi,\Lambda}\right)$
is the optimal negativity. 

Thus given a quantum channel $\Lambda$, the task is to find $\mathbb{F}\left(\Lambda\right)$
and $\mathcal{N}\left(\Lambda\right)$ and the protocols to achieve
these optimal values. Note that, it is quite possible that the optimal
values may be attained by sending different pure states. However,
the question that deserves utmost importance is whether the optimal
states are maximally entangled like noiseless channels. 

To the best of our knowledge, the problem concerning one-shot optimal
singlet fraction has been completely solved only in the qubit case
\cite{PBG-2014}. In particular, for any qubit channel (which is not
entanglement breaking), it was shown that $\left|\psi_{{\rm opt}}\right\rangle $,
satisfying (\ref{optimal-purity-one}) is maximally entangled if and
only if the channel is unital, and for any non-unital qubit channel
$\left|\psi_{{\rm opt}}\right\rangle $ is necessarily nonmaximally
entangled (for the specific case of amplitude damping channel; see
\cite{BG-2012}). Further, it was shown that for any qubit channel
$\Lambda_{{\rm qubit}}$, $\mathbb{F}\left(\Lambda_{{\rm qubit}}\right)$
can be exactly computed and is given by \cite{PBG-2014} 
\begin{eqnarray}
\mathbb{F}\left(\Lambda_{{\rm qubit}}\right) & = & \frac{1+2\mathcal{N}\left(\rho_{\Phi^{+},\Lambda_{{\rm qubit}}}\right)}{2}\label{F=00003Dnegatvity}
\end{eqnarray}
where $\left|\Phi^{+}\right\rangle =\frac{1}{\sqrt{2}}\left(\left|00\right\rangle +\left|11\right\rangle \right)$. 

In \cite{Ziman-Buzek-2005,Ziman-Buzek-2007}, specific examples were
given which showed that the ordering of entangled states may change
under one-sided local action of a qubit channel and the maximum output
entanglement may not be achieved for an input maximally entangled
state {[}shown for a system of four qubits having configuration $\left(\mathbb{C}^{2}\otimes\mathbb{C}^{2}\right)\otimes\left(\mathbb{C}^{2}\otimes\mathbb{C}^{2}\right)${]}.
A more systematic way supporting these observations can be found in
\cite{BG-2012,PBG-2014,Streltsov-etal-2015}. For example, in \cite{PBG-2014}
it was pointed out that for qubit channels, the maximum achievable
negativity may not be achieved by sending a maximally entangled state:
Using (\ref{F=00003Dnegatvity}) and (\ref{upper bound on optimal SF})
we see that 
\begin{equation}
\mathcal{N}\left(\rho_{\Phi^{+},\Lambda_{{\rm qubit}}}\right)\leq\mathcal{N}\left(\rho_{\psi_{{\rm opt}},\Lambda_{{\rm qubit}}}\right)\leq\mathcal{N}\left(\Lambda_{{\rm qubit}}\right)\label{inequality-Ngeativity}
\end{equation}
 Since $\left|\psi_{{\rm opt}}\right\rangle $ is nonmaximally entangled
for non-unital channels, the inequality implies that nonmaximally
entangled states also lead to maximum achievable entanglement negativity;
for an amplitude damping channel the inequality (\ref{inequality-Ngeativity})
is strict \cite{PBG-2014}. The question for other nonunital channels,
however, remains open. 

In this paper we extend our previous studies \cite{BG-2012,PBG-2014}
to higher dimensional quantum channels. In particular, we wish to
know whether we can find quantum channels in all higher dimensions
$d\geq3$ with properties similar to non-unital qubit channels. The
main results of this paper are the following. 
\begin{itemize}
\item We present a family of quantum channels $\Omega$ in every finite
dimension $d\geq3$ for which we prove that $\left|\psi_{{\rm opt}}\right\rangle $
is nonmaximally entangled. Although we are not able to provide an
expression for this optimal state, nonetheless, we obtain a nonmaximally
entangled state $\left|\psi^{\prime}\right\rangle \in\mathbb{C}^{d}\otimes\mathbb{C}^{d}$
satisfying the inequality: 
\begin{equation}
\mathbb{F}\left(\Omega\right)=\mathbb{F}^{*}\left(\rho_{\psi_{{\rm opt}},\Omega}\right)\geq\mathbb{F}\left(\rho_{\psi^{\prime},\Omega}\right)>\mathbb{F}^{*}\left(\rho_{\Phi,\Omega}\right)\label{main-result-inequality}
\end{equation}
where $\left|\psi^{\prime}\right\rangle $ is the eigenvector corresponding
to the largest eigenvalue of the density matrix $\rho_{\Phi^{+},\hat{\Omega}}$
with $\hat{\Omega}$ being the dual map (see the next section for
the definition) and $\left|\Phi\right\rangle \in\mathbb{C}^{d}\otimes\mathbb{C}^{d}$
being any maximally entangled state. Note that the first inequality
gives us a lower bound on the one-shot optimal singlet fraction and
shows that suitable nonmaximally entangled states are better than
maximally entangled states. Also note that, since $\mathbb{F}^{*}$
is a LOCC monotone, the above inequality together with the identity
(\ref{optimal-purity-one}) provides a constructive way to demonstrate
that ordering of $\mathbb{F}^{*}$ in general is not preserved under
TP-LOCC in all finite dimensions. 
\item Optimal negativity is attained only by appropriate nonmaximally entangled
states. Using (\ref{F*<=00003Dnegativity}), (\ref{optimal-purity-one})
and (\ref{eF(Omega)>Nagativity}) it is easy to see that 
\begin{eqnarray}
\mathcal{N}\left(\Omega\right)\geq\mathcal{N}\left(\rho_{\psi_{{\rm opt}},\Omega}\right) & > & \mathcal{N}\left(\rho_{\Phi^{+},\Omega}\right)\label{N(psi-opt)>N(maxentangled)}
\end{eqnarray}
where $\left|\psi_{{\rm opt}}\right\rangle $ is nonmaximally entangled.
Thus, in all finite dimensions $d\geq3$ we are able to show by explicit
construction that the maximum output entanglement, as measured by
negativity, is not always achieved using a maximally entangled input
state. This, significantly improves upon the previously known examples. 
\end{itemize}
We also make the following observation. We find that in higher dimensions
an expression analogous to (\ref{F=00003Dnegatvity}) does not hold
in general. This follows from inequality (see the proof of \ref{main-result-inequality}):
\begin{eqnarray}
\mathbb{F}\left(\Omega\right) & > & \frac{1+2\mathcal{N}\left(\rho_{\Phi^{+},\Omega}\right)}{d}\label{eF(Omega)>Nagativity}
\end{eqnarray}
where $\mathcal{N}\left(\rho_{\Phi^{+},\Omega}\right)$ is the negativity
of the density matrix $\rho_{\Phi^{+},\Omega}$. One may argue that
there is no convincing reason why one should have expected the generalization
to hold in the first place; however, the exact formula obtained in
\cite{PBG-2014} prompted us to think such a generalization, if it
holds, would give us a computable formula for one-shot optimal singlet
fraction in all finite dimensions. Unfortunately, our optimism turned
out to be misplaced.

\section{Results}

A quantum channel $\Lambda$ is a trace preserving completely positive
map characterized by a set of Kraus operators $\left\{ A_{i}\right\} $
satisfying $\sum A_{i}^{\dagger}A_{i}=\mathcal{I}$ (see for example,
\cite{Schumacher-1996}). The dual map $\hat{\Lambda}$, described
in terms of the Kraus operators $\left\{ A_{i}^{\dagger}\right\} $,
is the adjoint map with respect to the Hilbert-Schmidt inner product.
We say that a channel $\Lambda$ is \emph{unital} if its action preserves
the Identity: $\Lambda\left(\mathcal{I}\right)=\mathcal{I}$, and
\emph{nonunital} if it does not, i.e., $\Lambda\left(\mathcal{I}\right)\neq\mathcal{I}$.
Moreover, the dual map $\hat{\Lambda}$ is trace-preserving, and hence
a channel, iff $\Lambda$ is unital. The one-sided action of a $d$-dimensional
map $\$\in\left\{ \Lambda,\hat{\Lambda}\right\} $ on a pure state
$\vert\psi\rangle\in\mathbb{C}^{d}\otimes\mathbb{C}^{d}$ gives rise
to a mixed state which can be conveniently expressed as: 
\begin{eqnarray*}
\rho_{\psi,\$} & = & \left(\mathcal{I}\otimes\$\right)\rho_{\psi}\\
 & = & \sum_{i}\left(\mathcal{I}\otimes K_{i}\right)\rho_{\psi}\left(\mathcal{I}\otimes K_{i}^{\dagger}\right)
\end{eqnarray*}
where the Kraus operators $\left\{ K_{i}\right\} $ describe the channel
$\$$ and $\rho_{\psi}=\left|\psi\right\rangle \left\langle \psi\right|$
is the density matrix corresponding to the pure state $\left|\psi\right\rangle $.
We now give two useful lemmas which are applicable to any quantum
channel $\Lambda$. The first lemma was proved in \cite{Ziman-Buzek-2005}. 
\begin{lem}
\label{F*phi=00003DF*phi+} For a $d$-dimensional quantum channel
$\Lambda$, $\mathbb{F}^{*}\left(\rho_{\Phi,\Lambda}\right)=\mathbb{F}^{*}\left(\rho_{\Phi^{+},\Lambda}\right)$
where $\left|\Phi^{+}\right\rangle =\frac{1}{\sqrt{d}}\sum_{i=0}^{d-1}\left|ii\right\rangle $
and $\left|\Phi\right\rangle $ is any maximally entangled state in
$\mathbb{C}^{d}\otimes\mathbb{C}^{d}$. 
\end{lem}
The proof is simple. Since every maximally entangled state $\left|\Phi\right\rangle \in\mathbb{C}^{d}\otimes\mathbb{C}^{d}$
can be written as $\left|\Phi\right\rangle =\left(U\otimes V\right)\left|\Phi^{+}\right\rangle $
for some $U,V\in SU\left(d\right)$, using the identity $\left(\mathcal{I}\otimes V\right)\left|\Phi^{+}\right\rangle =\left(V^{T}\otimes\mathcal{I}\right)\left|\Phi^{+}\right\rangle $
we can write $\left|\Phi\right\rangle =\left(W\otimes\mathcal{I}\right)\left|\Phi^{+}\right\rangle $
where $W=UV^{T}$ is also a unitary operator. Because the channel
$\Lambda$ acts only on the second qudit, we have $\rho_{\Phi,\Lambda}=\left(W\otimes\mathcal{I}\right)\rho_{\Phi^{+},\Lambda}\left(W^{\dagger}\otimes\mathcal{I}\right).$
Thus the density matrices $\rho_{\Phi,\Lambda}$ and $\rho_{\Phi^{+},\Lambda}$
are connected by a local unitary operator acting on the first system.
Because the first system never interacts with the channel, this local
unitary can always be absorbed in the post-transmission optimal TP-LOCC
associated with the state transformations (defined earlier) $\rho_{\Phi,\Lambda}\rightarrow\rho_{\Phi,\Lambda}^{\prime}$
and $\rho_{\Phi^{+},\Lambda}\rightarrow\rho_{\Phi^{+},\Lambda}^{\prime}$.
Therefore, $\mathbb{F}^{*}\left(\rho_{\Phi,\Lambda}\right)=\mathbb{F}^{*}\left(\rho_{\Phi^{+},\Lambda}\right)$. 
\begin{lem}
\label{F(lambda)>=00003Dlambdamax} For a $d$-dimensional quantum
channel $\Lambda$, $\mathbb{F}\left(\Lambda\right)\geq\lambda_{\max}\left(\rho_{\Phi^{+},\Lambda}\right)$
where $\left|\Phi^{+}\right\rangle =\frac{1}{\sqrt{d}}\sum_{i=0}^{d-1}\left|ii\right\rangle $
and $\lambda_{\max}\left(\rho_{\Phi^{+},\Lambda}\right)$ is the largest
eigenvalue of the density matrix $\rho_{\Phi^{+},\Lambda}$. \end{lem}
\begin{proof}
The proof is along the same lines as in the qubit case \cite{PBG-2014}.
We begin by noting that for any $\left|\psi\right\rangle \in\mathbb{C}^{d}\otimes\mathbb{C}^{d}$,
\begin{eqnarray}
\mathbb{F}\left(\Lambda\right) & \geq & \max_{\psi}\mathbb{F}\left(\rho_{\psi,\Lambda}\right)=\max_{\psi,\Phi}\left\langle \Phi\right|\rho_{\psi,\Lambda}\left|\Phi\right\rangle \label{F-Lambda-inequality}
\end{eqnarray}
where $\left|\Phi\right\rangle $ is maximally entangled. Using the
relations $\left|\Phi\right\rangle =\left(U\otimes V\right)\left|\Phi^{+}\right\rangle $
for some $U,V\in SU\left(d\right)$ and $\left(\mathcal{I}\otimes V\right)\left|\Phi^{+}\right\rangle =\left(V^{T}\otimes\mathcal{I}\right)\left|\Phi^{+}\right\rangle $,
it is straightforward to show that 
\begin{eqnarray}
\mathbb{F}\left(\rho_{\psi,\Lambda}\right) & = & \left\langle \psi\left|\rho_{\Phi^{+},\hat{\Lambda}}\right|\psi\right\rangle \label{eq:F-rhhopsilambda}
\end{eqnarray}
where $\hat{\Lambda}$ is the dual channel. From (\ref{F-Lambda-inequality})
and (\ref{eq:F-rhhopsilambda}) we therefore get
\begin{eqnarray*}
\mathbb{F}\left(\Lambda\right) & \geq & \lambda_{\max}\left(\rho_{\Phi^{+},\hat{\Lambda}}\right)\\
 & = & \lambda_{\max}\left(\rho_{\Phi^{+},\Lambda}\right)
\end{eqnarray*}
where we have used $\lambda_{\max}\left(\rho_{\Phi^{+},\hat{\Lambda}}\right)=\lambda_{\max}\left(\rho_{\Phi^{+},\Lambda}\right)$
proved in \cite{PBG-2014} for any $d$ dimensional channel $\Lambda$.
\end{proof}

\subsection*{Main results}

Let us now consider the $d$-dimensional quantum channel $\Omega$
defined by the Kraus operators $A_{i}$ for $i=0,\dots,d-1$,
\begin{eqnarray}
A_{0}={\rm diag}\left(1,x_{1},x_{2},\dots,x_{d-1}\right) & ; & \left(A_{m}\right)_{ij}=\sqrt{1-x_{m}^{2}}\delta_{0i}\delta_{mj}\; i,j=0,\dots d-1\;\forall\, m=1,\dots,d-1\label{KraussOmega}
\end{eqnarray}
where $0<x_{i}<1$ for every $i$ and $x_{i}\neq x_{j}$ for at least
one pair $\left(i,j\right)$. That the Kraus operators defined above
indeed describe a legitimate quantum channel can be seen as follows.
First, it is easy to check that
\begin{eqnarray}
\left(A_{m}^{\dagger}A_{m}\right)_{ik}=\left(1-x_{m}^{2}\right)\delta_{mi}\delta_{mk} & ; & A_{0}^{\dagger}A_{0}={\rm diag}\left(1,x_{1}^{2},x_{2}^{2},\dots,x_{d-1}^{2}\right)\label{AdaggerA}
\end{eqnarray}
Clearly the operators $A_{i}^{\dagger}A_{i}$ are positive and moreover,
Eqs.$\,$(\ref{AdaggerA}) lead to 
\begin{eqnarray}
A_{0}^{\dagger}A_{0}+\sum_{m=1}^{d-1}A_{m}^{\dagger}A_{m} & = & \mathcal{I}_{d\times d}.\label{completeness-Kraus}
\end{eqnarray}
We now state our result. 
\begin{thm}
For the $d$-dimensional quantum channel $\Omega$ described above,
the following inequalities hold in every finite dimension $d\geq3$:
\begin{equation}
\mathbb{F}\left(\Omega\right)\geq\mathbb{F}\left(\rho_{\psi^{\prime},\Omega}\right)>\mathbb{F}^{*}\left(\rho_{\Phi,\Omega}\right)\label{Inequality-Theorem-1}
\end{equation}
 where $\left|\psi^{\prime}\right\rangle \in\mathbb{C}^{d}\otimes\mathbb{C}^{d}$
is a pure state, not maximally entangled, and $\left|\Phi\right\rangle \in\mathbb{C}^{d}\otimes\mathbb{C}^{d}$
is any maximally entangled state. 
\end{thm}
The inequalities (\ref{Inequality-Theorem-1}) are established through
the following results. 
\begin{lem}
\label{lambdamax>negativity} For any maximally entangled state $\left|\Phi\right\rangle \in\mathbb{C}^{d}\otimes\mathbb{C}^{d}$,
\begin{eqnarray}
\lambda_{\max}\left(\rho_{\Phi^{+},\Omega}\right) & > & \frac{1+2\mathcal{N}\left(\rho_{\Phi^{+},\Omega}\right)}{d}\label{Lemma3-inequality}
\end{eqnarray}
 for all $d\geq3$, \end{lem}
\begin{proof}
First we obtain $\lambda_{\max}\left(\rho_{\Phi^{+},\Omega}\right)$.
The action of the Kraus operators given by (\ref{KraussOmega}) on
$\left|\Phi^{+}\right\rangle $ are given by: 
\begin{eqnarray}
\left(\mathcal{I}\otimes A_{0}\right)\left|\Phi^{+}\right\rangle  & = & \frac{1}{\sqrt{d}}\left(\left|00\right\rangle +\sum_{i=1}^{d-1}x_{i}\left|i\right\rangle \left|i\right\rangle \right)=\left|\phi_{0}\right\rangle ,\label{phi0}
\end{eqnarray}
and for $m=1,\dots,d-1$ 
\begin{eqnarray}
\left(\mathcal{I}\otimes A_{m}\right)\left|\Phi^{+}\right\rangle  & = & \frac{1}{\sqrt{d}}\sum_{i=0}^{d-1}\left|i\right\rangle A_{m}\left|i\right\rangle \nonumber \\
 & = & \frac{1}{\sqrt{d}}\sum_{i=0}^{d-1}\left|i\right\rangle \sqrt{1-x_{m}^{2}}\delta_{im}\left|0\right\rangle \;\because A_{m}\left|i\right\rangle =\sqrt{1-x_{m}^{2}}\delta_{im}\left|0\right\rangle \nonumber \\
 & = & \frac{1}{\sqrt{d}}\sqrt{1-x_{m}^{2}}\left|m\right\rangle \left|0\right\rangle =\left|\phi_{m}\right\rangle .\label{phim}
\end{eqnarray}
Thus, 
\begin{eqnarray}
\rho_{\Phi^{+},\Omega} & = & \sum_{m=0}^{d-1}\left(\mathcal{I}\otimes A_{m}\right)\rho_{\Phi^{+}}\left(\mathcal{I}\otimes A_{m}^{\dagger}\right)\nonumber \\
 & = & \left|\phi_{0}\right\rangle \left\langle \phi_{0}\right|+\sum_{m=1}^{d-1}\left|\phi_{m}\right\rangle \left\langle \phi_{m}\right|.\label{rhophi+omega}
\end{eqnarray}
As $\rho_{\Phi^{+},\Omega}$ is already in the diagonal form, it is
straightforward to obtain its largest eigenvalue,
\begin{eqnarray}
\lambda_{\max}\left(\rho_{\Phi^{+},\Omega}\right) & = & \frac{1}{d}\left(1+\sum_{i=1}^{d-1}x_{i}^{2}\right).\label{lambdamax}
\end{eqnarray}
Next, we compute negativity $\mathcal{N}\left(\rho_{\Phi^{+},\Omega}\right)$.
The partial transposed matrix corresponding to $\rho_{\Phi^{+},\Omega}$
is given by
\begin{eqnarray*}
\rho_{\Phi^{+},\Omega}^{\Gamma} & = & \frac{1}{d}\left[\left|00\right\rangle \left\langle 00\right|+\sum_{i=1}^{d-1}x_{i}\left(\left|0i\right\rangle \left\langle i0\right|+\left|i0\right\rangle \left\langle 0i\right|\right)+\sum_{i,j=1}^{d-1}x_{i}x_{j}\left|ij\right\rangle \left\langle ji\right|+\sum_{i=1}^{d-1}\left(1-x_{i}^{2}\right)\left|i0\right\rangle \left\langle i0\right|\right]
\end{eqnarray*}
with easily computed eigenvalues,
\begin{eqnarray*}
\frac{1}{d}\;\left({\rm multiplicity}\, d\right)\;; & \;\pm\frac{x_{i}^{2}}{d},\; i=1,\dots,d-1\;; & \;\pm\frac{x_{i}x_{j}}{d},\; i<j\; i,j=1,\dots,d-1.
\end{eqnarray*}
 As negativity is defined as the absolute value of the sum of the
negative eigenvalues \cite{Negativity}, we have 
\begin{eqnarray}
\mathcal{N}\left(\rho_{\Phi^{+},\Omega}\right) & = & \frac{1}{d}\left(\sum_{i=1}^{d-1}x_{i}^{2}+\sum_{1\leq i<j\leq d-1}x_{i}x_{j}\right).\label{negativity}
\end{eqnarray}
From (\ref{lambdamax}) and (\ref{negativity}) we see that that the
inequality (\ref{Lemma3-inequality}) holds provided:
\begin{eqnarray*}
\left(d-2\right)\sum_{i=1}^{d-1}x_{i}^{2} & > & 2\sum_{1\leq i<j\leq d-1}x_{i}x_{j}.
\end{eqnarray*}
Now, $\sum_{1\leq i<j\leq d-1}\left(x_{i}-x_{j}\right)^{2}>0$, since
for at least one pair $\left(i,j\right)$, $x_{i}\neq x_{j}$ (as
given in the definition of the channel), the above inequality always
holds for all $d\geq3$. This completes the proof. 
\end{proof}
Let us now note the consequences of the above lemma. 

Since $\mathbb{F}\left(\Omega\right)\geq\lambda_{\max}\left(\rho_{\Phi^{+},\Omega}\right)$
(from Lemma \ref{F(lambda)>=00003Dlambdamax}), we see that 
\begin{eqnarray*}
\mathbb{F}\left(\Omega\right) & > & \frac{1+2\mathcal{N}\left(\rho_{\Phi^{+},\Omega}\right)}{d}.
\end{eqnarray*}
Thus the generalization of the formula (\ref{F=00003Dnegatvity})
that allows us to compute optimal fidelity exactly for qubit channels
does not hold in general in higher dimensions. 

We now show that $\left|\psi_{{\rm opt}}\right\rangle $ is a nonmaximally
entangled state. From Eq.$\,$(\ref{F*<=00003Dnegativity}), we have
$\mathbb{F}^{*}\left(\rho_{\Phi^{+}.\Omega}\right)\leq\frac{1+2\mathcal{N}\left(\rho_{\Phi^{+},\Omega}\right)}{d}$.
Using this inequality, and (\ref{F(lambda)>=00003Dlambdamax}), and
the inequality (\ref{Lemma3-inequality}) we immediately obtain $\mathbb{F}\left(\Omega\right)\geq\lambda_{\max}\left(\rho_{\Phi^{+},\Omega}\right)>\mathbb{F}^{*}\left(\rho_{\Phi^{+}.\Omega}\right)$.
Since $\mathbb{F}^{*}\left(\rho_{\Phi^{+}.\Omega}\right)=\mathbb{F}^{*}\left(\rho_{\Phi.\Omega}\right)$
for any maximally entangled state $\left|\Phi\right\rangle \in\mathbb{C}^{d}\otimes\mathbb{C}^{d}$
{[}Lemma \ref{F*phi=00003DF*phi+}{]}, we have 
\begin{equation}
\mathbb{F}\left(\Omega\right)\geq\lambda_{\max}\left(\rho_{\Phi^{+},\Omega}\right)>\mathbb{F}^{*}\left(\rho_{\Phi.\Omega}\right).\label{intermediate-inequality}
\end{equation}
Noting that $\mathbb{F}\left(\Omega\right)=\mathbb{F}^{*}\left(\rho_{\psi_{{\rm opt}},\Omega}\right)$,
we get $\mathbb{F}^{*}\left(\rho_{\psi_{{\rm opt}},\Omega}\right)>\mathbb{F}^{*}\left(\rho_{\Phi.\Omega}\right)$
for all maximally entangled states $\left|\Phi\right\rangle \in\mathbb{C}^{d}\otimes\mathbb{C}^{d}$.
We therefore conclude that $\left|\psi_{{\rm opt}}\right\rangle $
must be nonmaximally entangled. While we are unable to obtain $\left|\psi_{{\rm opt}}\right\rangle $,
the following lemma gives us a possible candidate and allows us to
obtain a lower bound on $\mathbb{F}\left(\Omega\right)$. 
\begin{lem}
\label{existence of psiprime} Let $\left|\psi^{\prime}\right\rangle $
be the eigenvector corresponding to the eigenvalue $\lambda_{{\rm max}}\left(\rho_{\Phi^{+},\hat{\Omega}}\right)$.
Then, $\lambda_{\max}\left(\rho_{\Phi^{+},\Omega}\right)=\mathbb{F}\left(\rho_{\psi^{\prime},\Omega}\right)$.
Moreover, $\left|\psi^{\prime}\right\rangle $ is not maximally entangled. \end{lem}
\begin{proof}
From Eq.$\,$(\ref{eq:F-rhhopsilambda}) we know that for any pure
state $\left|\psi\right\rangle \in\mathbb{C}^{d}\otimes\mathbb{C}^{d}$,
$\mathbb{F}\left(\rho_{\psi,\Omega}\right)=\left\langle \psi\left|\rho_{\Phi^{+},\hat{\Omega}}\right|\psi\right\rangle $.
As $\left|\psi^{\prime}\right\rangle $ is the eigenvector corresponding
to the eigenvalue $\lambda_{{\rm max}}\left(\rho_{\Phi^{+},\hat{\Omega}}\right)$,
this means, 
\[
\lambda_{{\rm max}}\left(\rho_{\Phi^{+},\hat{\Omega}}\right)=\left\langle \psi^{\prime}\left|\rho_{\Phi^{+},\hat{\Omega}}\right|\psi^{\prime}\right\rangle =\mathbb{F}\left(\rho_{\psi^{\prime},\Omega}\right).
\]
Using the identity $\lambda_{\max}\left(\rho_{\Phi^{+},\hat{\Lambda}}\right)=\lambda_{\max}\left(\rho_{\Phi^{+},\Lambda}\right)$
\cite{PBG-2014} for any quantum channel $\Lambda$, we therefore
have 
\begin{eqnarray*}
\lambda_{\max}\left(\rho_{\Phi^{+},\Omega}\right) & = & \mathbb{F}\left(\rho_{\psi^{\prime},\Omega}\right).
\end{eqnarray*}
On the other hand we have already shown that $\lambda_{\max}\left(\rho_{\Phi^{+},\Omega}\right)>\mathbb{F}^{*}\left(\rho_{\Phi,\Omega}\right)$.
Therefore, $\mathbb{F}\left(\rho_{\psi^{\prime},\Omega}\right)>\mathbb{F}^{*}\left(\rho_{\Phi,\Omega}\right)$
for any maximally entangled state $\left|\Phi\right\rangle $ from
which we conclude that $\left|\psi^{\prime}\right\rangle $ is not
a maximally entangled state. 
\end{proof}
Inequalities (\ref{intermediate-inequality}) and Lemma \ref{existence of psiprime}
conclude the proof of the theorem.

\section{Conclusions}

For any given $d$-dimensional quantum channel $\Lambda$ with $d\geq2$,
its one-shot optimal singlet fraction $\mathbb{F}\left(\Lambda\right)$
defines the maximum singlet fraction achievable for entangled states
established between two remote observers for every use of the channel.
Recall that 
\begin{eqnarray}
\mathbb{F}\left(\Lambda\right) & = & \mathbb{F}^{*}\left(\rho_{\psi_{{\rm opt}},\Lambda}\right)=\mathbb{F}\left(\rho_{\psi_{{\rm opt}},\Lambda}^{\prime}\right).\label{optimal-purity-one-1}
\end{eqnarray}
Thus, $\mathbb{F}\left(\Lambda\right)$ quantifies how useful a channel
$\Lambda$ is either for direct applications for quantum information
processing tasks, e.g. teleportation \cite{Horodecki-general-teleportation-1999}
or for entanglement distillation where the yield depends upon the
singlet fraction of the noisy states. 

For qubit channels $\mathbb{F}\left(\Lambda\right)$ can be exactly
computed and the relevant questions have been satisfactorily answered
before \cite{PBG-2014}. The results, however, point towards two counter-intuitive
features. Foremost among them is that $\left|\psi_{{\rm opt}}\right\rangle $
is nonmaximally entangled if and only if the channel is nonunital.
And the next is, for nonunital qubit channels maximum achievable entanglement
negativity using a maximally entangled state cannot be more than what
is attained by sending $\left|\psi_{{\rm opt}}\right\rangle $. In
fact, for an amplitude damping channel (a nonunital channel) it was
further shown that optimal negativity is obtained only by a nonmaximally
entangled state. 

Motivated by the above results we wanted to understand how well the
results and observations made for qubit channels hold in higher dimensions.
We presented a family of qudit channels $\Omega$ in all finite dimensions
$d\geq3$ for which we proved properties similar to nonunital qubit
channels. In particular, we proved that one-shot optimal singlet fraction
and negativity are attained only using appropriate nonmaximally entangled
states. However, we also find that a generalized version of the formula
that allows us to compute the optimal singlet fraction exactly for
qubit channels does not hold in general in higher dimensions. 

While a lot of results had been obtained characterizing quantum channels,
we believe that much less is understood when it comes to characterizing
quantum channels through the notions of optimal singlet fraction and
entanglement measures. In higher dimensions almost every interesting
question is left open, and probably a good way to address them is
to solve the questions for specific channels of interest e.g. a depolarizing
channel. Such results can provide us with useful insights. Another
paradigm within which where we can ask similar questions is entanglement
distribution in the presence of preshared correlations. 
\begin{acknowledgments}
Part of this work was completed when Rajarshi Pal was a long term
visitor at Bose Institute during May 2015- March 2016. S.B. is supported
in part by SERB (Science and Engineering Research Board), DST, Govt.
of India--Project No. EMR/2015/002373. \end{acknowledgments}

\end{document}